\documentclass[a4paper,11pt]{article}
\usepackage{graphicx,color}
\usepackage{epsfig}
\usepackage{amsmath}
\usepackage{bbm}
\usepackage{amsfonts}
\usepackage{amssymb}
\usepackage{amsthm}
\usepackage{verbatim}
\usepackage{cite}

\newtheorem{theorem}{Theorem}

\newcommand{\ket}[1]{\left\vert{#1}\right\rangle}
\newcommand{\bra}[1]{\left\langle{#1}\right\vert}
\newcommand{\sprod}[2]{\left\langle{#1}|{#2}\right\rangle}
\newcommand{\scprod}[3]{\langle{#1}|{#2}|{#3}\rangle}

\begin{document}

\title{Orbits of Mutually Unbiased Bases}
\date{}
\author{Kate Blanchfield\footnote{kate@fysik.su.se} \\ {\it Stockholms universitet, Fysikum, S-106 91 Stockholm, Sweden}}

\maketitle

\begin{abstract}
We express Alltop's construction of mutually unbiased bases as orbits under the Weyl-Heisenberg group in prime dimensions and find a related construction in dimensions 2 and 4. We reproduce Alltop's mutually unbiased bases using abelian subgroups of the Clifford group in prime dimensions, in direct analogy to the well-known construction of mutually unbiased bases using abelian subgroups of the Weyl-Heisenberg group. Finally, we prove three theorems relating to the distances and linear dependencies among different sets of mutually unbiased bases.
\end{abstract}

\section{Introduction}

Mutually unbiased bases (MUBs) provide the mathematical formalism behind Bohr's idea of complementarity in quantum mechanics. Complementary observables, as developed by Schwinger \cite{Schwinger}, have eigenbases that are mutually unbiased. In finite dimension $N$, we express this via the condition
\begin{equation}
| \left\langle e_i | f_j \right\rangle |^2 = \frac{1}{N} 
\label{eq:MUB}
\end{equation}
where the vectors $\ket{e_i}$ come from one basis and the vectors $\ket{f_j}$ from another basis. Consequently, if we prepare a quantum state in the first basis and perform a measurement in the second basis, each outcome is equally likely. This makes MUBs useful for a host of practical reasons including quantum cryptography \cite{crypto0, crypto1, crypto2}, state tomography \cite{WF, tomo1, tomo2} and entanglement detection \cite{ent}. A comprehensive review of MUBs can be found in \cite{review}.

One open question about MUBs is how many exist in non-prime power dimensions. If the dimension $N$ is a prime or prime power, we can always find $N+1$ MUBs. This is often called a complete set because it is the maximum number possible \cite{WF} and we shall use the term `complete MUB' to denote such a set. In composite dimensions, we don't know the maximum number of MUBs although a combination of numerical \cite{Stefan1, Berge}, analytical \cite{Jaming, Stefan2} and computer-algebraic \cite{Grassl, Stefan3} work strongly suggests the maximum number in dimension 6 is three.

We can also ask whether different complete MUBs exist. In dimension $N \leq 5$ all complete MUBs are unitarily equivalent \cite{russian, Ingemar2}, but unitarily inequivalent complete MUBs are known in higher dimensions \cite{Kantor}. 
The main focus of this paper is the construction of complete MUBs using cubic functions introduced by Alltop in 1980 \cite{Alltop}. Alltop's original paper was motivated by classical signal processing, but the functions were shown to lead to complete MUBs in prime and odd prime power dimensions \cite{KR2}. Recently, a new family of Alltop functions was found \cite{Hall}.

One might take the attitude that once we have one complete MUB in a given dimension we are not interested in constructing any unitarily equivalent others. Our opinion is that additional constructions---and their resulting complete MUBs---are still interesting as they reveal new structure. For example, different complete sets of MUBs in dimensions 8 and 16 exhibit different amounts of entanglement \cite{Gunnar} and Alltop's MUBs relate to the Clifford group in prime dimensions in a way that other complete MUBs do not. The new constructions also open avenues between different areas of research: the Alltop MUBs we study in this paper appear in generalisations of the `pi over eight' gate, needed for universal quantum computing \cite{Howard}, as well as settling an open question concerning another special quantum measurement called a SIC.

In Section 2 of this paper we introduce relevant groups for the MUB problem: the Weyl-Heisenberg group and the Clifford group. Section 3 looks at complete MUBs in prime dimensions, focusing on the Alltop construction, and Section 4 relates these complete MUBs to the Clifford group. Though Alltop's construction is already known, we hope to present it here in a way we think is particularly simple and which brings it into line with the usual construction by Ivanovi\'{c} \cite{Ivanovic} and Wootters and Fields \cite{WF}. Section 5 investigates distances among bases in complete MUBs from both constructions. We use their status as 2-designs to prove that these bases lie at regular distances from one another when viewed as points in a Grassmannian space. Section 6 looks at linear dependencies among MUB vectors and Section 7 introduces complete MUBs in dimensions $N=2$ and 4 that have a group structure analogous to those from Alltop's construction in prime dimensions. We summarise our results in Section 8. Finally, the appendix links the Alltop MUBs to an open question about SICs in dimension 3.

\section{The Weyl-Heisenberg and Clifford groups}

This section introduces the Weyl-Heisenberg group and the Clifford group. We will work with MUBs in finite dimension $N$ and, with the exception of Section 8, will always use prime $N$. 

The Weyl-Heisenberg (WH) group is generated by the operators $X$ and $Z$ whose action on the computational basis is
\begin{eqnarray}
X \ket{a} = \ket{a+1}
\label{eq:X} \\
Z \ket{a} = \omega \ket{a}
\label{eq:Z}
\end{eqnarray}
where $\omega = e^{\frac{2 \pi i}{N}}$ and addition is modulo $N$. In dimension 2, these are the familiar Pauli spin matrices. We can write a general group element as a displacement operator
\begin{equation}
D_{\bf p} = \tau^{p_1 p_2} X^{p_1} Z^{p_2}
\end{equation}
where $\tau = - e^{\frac{\pi i}{N}}$ and ${\bf p}$ is a 2-component vector whose entries $p_1$ and $p_2$ lie in $\mathbbm{Z}_N \times \mathbbm{Z}_N$, where $\mathbbm{Z}_N$ are integers modulo $N$. We ignore phase factors, which leaves $N^2$ elements in the WH group.

The Clifford group is the group of unitary operations that stabilises the WH group. Disregarding complications with phases, it is isomorphic to the semi-direct product of the WH group and the symplectic group $SL(2,\mathbbm{Z}_{N})$.

The symplectic group consists of all matrices
\begin{equation}
\mathcal{G} = \begin{pmatrix}
\alpha & \beta \\
\gamma & \delta \\
\end{pmatrix} 
\end{equation}
with $\alpha, \beta, \gamma, \delta \in \mathbbm{Z}_N$ and determinant 1 (mod $N$). For each element $\mathcal{G}$, there is a corresponding unitary $U_{\mathcal{G}}$ in the Clifford group. This unitary representation is given by 
\begin{equation}
U_{\mathcal{G}} = \frac{e^{i \theta}}{\sqrt{N}} \sum_{u,v=0}^{N-1} \tau^{\beta^{-1} ( \delta u^2 - 2uv + \alpha v^2 )}
\label{eq:marcus}
\end{equation}
where $e^{i \theta}$ is an arbitrary phase to be determined by the order or $U_{\mathcal{G}}$ \cite{Marcus}. If $\beta$ does not have a multiplicative inverse then we must use a decomposition
\begin{equation}
\mathcal{G} = \mathcal{G}_1 \mathcal{G}_2 = 
\begin{pmatrix}
\alpha_1 & \beta_1 \\
\gamma_1 & \delta_1 
\end{pmatrix} 
\begin{pmatrix}
\alpha_2 & \beta_2 \\
\gamma_2 & \delta_2 
\end{pmatrix} 
\end{equation}
together with Eq. (\ref{eq:marcus}) to calculate $U_{\mathcal{G}}$ from
\begin{equation}
U_{\mathcal{G}} = U_{\mathcal{G}_1} U_{\mathcal{G}_2} 
\end{equation}
up to an overall phase \cite{Marcus}. We can express the action of the Clifford unitaries on the WH group as
\begin{equation}
U_{\mathcal{G}}^{\dagger} D_{\bf p} U_{\mathcal{G}} = D_{\mathcal{G} {\bf p}} .
\end{equation}
Again ignoring phases, any Clifford group element can be expressed as 
\begin{equation}
D_{\bf p} U_{\mathcal{G}} .
\label{eq:clifford_element}
\end{equation}
We shall be interested in elements of order $N$, which requires $U_{\mathcal{G}}$ to also be order $N$. For prime $N$, Sylow's theorems state that if $N$ divides the order of a group we are guaranteed subgroups of order $N$. The order of the symplectic group is $N(N^2-1)$ and we find $N+1$ order $N$ subgroups. This gives $(N+1)(N-1)$ elements of the form $U_{\mathcal{G}}$. Recall that there are $N^2$ displacement operators. Combining these gives $N^2(N^2-1)$ order $N$ Clifford elements (ignoring the elements that also appear in the WH group, i.e. elements with $U_{\mathcal{G}}=\mathbbm{1}$), or, equivalently, $N^2(N+1)$ order $N$ subgroups.

A subset of the Clifford group can be written as Weyl-Heisenberg translates of the form
\begin{equation}
D_{\bf p} U_{\mathcal{G}} D_{\bf p}^{-1} .
\end{equation}
The translates have the same degenerate spectra as $U_{\mathcal{G}}$. We can count how many translates exist in the Clifford group. If we pick the matrix
\begin{equation}
\mathcal{S} = 
\left(
\begin{array}{cc}
1 & 0 \\
1 & 1 \\
\end{array}
\right)
\end{equation}
then the corresponding unitary has matrix components given by
\begin{equation}
\left( U_{\mathcal{S}} \right)_{mn} = \omega^{\frac{m^2}{2}} \delta_{mn} .
\label{eq:S}
\end{equation}
Note that the factor $\frac{1}{2}$ denotes the multiplicative inverse of 2 in $\mathbbm{Z}_N$. The unitary $U_{\mathcal{S}}$ is diagonal and commutes with $Z$, meaning the Clifford element $Z U_{\mathcal{S}} Z^{-1}$ is equal to $U_{\mathcal{S}}$. Only the translates involving $X$ give a new Clifford element, of which there are $N$ for each choice of $U_{\mathcal{S}}$ (as we are also interested in the cases where $D_{\bf p}=\mathbbm{1}$). In total, we find $N(N^2-1)$ WH translates.

This leaves $N(N-1)(N^2-1)$ Clifford elements of order $N$ that cannot be written as translates. We shall see in Section 4 that these play a crucial role in constructing complete MUBs from Alltop's functions.

\section{Orbit complete MUBs from the Weyl-Heisenberg group}

Ivanovi\'{c} published the first explicit construction of MUBs in prime dimensions \cite{Ivanovic}. We quickly go through his construction in the language of Bandyopadhyay {\it et al} \cite{B}. This method partitions the WH group into $N+1$ maximally abelian subgroups and then takes the joint eigenbasis defined by each subgroup. These eigenbases are mutually unbiased and so form a complete MUB.

In our representation of the WH group the eigenbasis of $Z$ is the computational basis. There is no unitary operation that relates all the $N+1$ bases---often called a MUB cycler---in odd prime dimensions\footnote{In dimension $N=4k+3$, there is an anti-unitary that does the job \cite{Marcus_09}.}, but the combination of two Clifford unitaries can do it: the Fourier matrix $U_F$, defined by
\begin{equation}
(U_F)_{mn} = \frac{1}{\sqrt{N}} \omega^{mn} ,
\label{eq:fourier}
\end{equation}
and the order $N$ unitary $U_{\mathcal{S}}$ defined in Equation (\ref{eq:S}). Acting with $U_F$ rotates the computational basis into a mutually unbiased basis (sometimes called the Fourier basis) while repeatedly acting with $U_{\mathcal{S}}$ on this second basis rotates it into the remaining $N-1$ mutually unbiased bases. This complete MUB forms an orbit under the Clifford group and we refer to it as the standard complete MUB. 

The construction arising from Alltop's cubic functions \cite{Alltop} produces a complete MUB in prime dimensions, which is unitarily equivalent to the standard one. Alltop gave an explicit expression for a fiducial vector whose orbit under the WH group collects into $N$ MUBs. They are also mutually unbiased to the computational basis, so together with this basis the orbit forms a complete MUB. We then have two complete MUBs---the standard one and the orbit one---that both include the computational basis. We shall sometimes say that they overlap at the computational basis.

Alltop's construction results in further orbit complete MUBs. Firstly, the expression for the fiducial vector can be modified to produce additional fiducial vectors. The expression is then
\begin{eqnarray}
\ket{f_x} = \frac{1}{\sqrt{N}} \sum_{a=0}^{N-1} \sigma^{a x^3} \ket{a} \quad \mbox{ for} \quad N=3
\label{eq:fiducial_3} \\
\ket{f_x} = \frac{1}{\sqrt{N}} \sum_{a=0}^{N-1} \omega^{x a^3} \ket{a} \quad \mbox{ for} \quad N>3
\label{eq:fiducial}
\end{eqnarray}
where $\sigma=e^{\frac{2 \pi i}{9}}$, $\omega=e^{\frac{2 \pi i}{N}}$ and $x = \{1, \ldots, N-1 \}$. Note the different position of the cubic power in the two fiducials. This means that there are $N-1$ fiducial vectors, labelled by $x$, that each produce an orbit complete MUB under the action of the WH group when combined with the computational basis. In other words, we now have $N-1$ orbit complete MUBs that overlap at the computational basis.
Secondly, the same rotations that permuted the bases within the standard complete MUB can be used to permute entire orbit complete MUBs into new orbit complete MUBs. This generates $(N-1)(N+1)=N^2-1$ orbit complete MUBs in total using $N(N^2-1)$ bases (excluding those found in the standard complete MUB). Their behaviour under the Clifford group depends on dimension. When $N=3$ or $N=3k+2$, the orbit complete MUBs form a single orbit under the Clifford group. When $N=3k+1$, they split into 3 Clifford orbits \cite{HC}.

There is an order $N$ unitary matrix not in the Clifford group that relates $N$ of these complete MUBs---always including the standard one. If we begin with the standard complete MUB, we can write down a diagonal unitary that takes the bases into those in the first orbit complete MUB. Its entries come from Equation (\ref{eq:fiducial}), i.e.
\begin{equation}
\begin{pmatrix}
1 &  &  &  \\
 & \omega &  &  \\
 &  & \omega^{2^3} &  \\
 &  &  & \omega^{3^3} &  \\
 &  &  &  &\ddots \\
\end{pmatrix} .
\end{equation}
The matrix will cycle through each orbit complete MUB that overlaps at the computational basis until, after $N$ applications, it returns to the standard complete MUB. We can similarly relate complete MUBs that overlap at other bases. The situation in dimension 3 is slightly different, as we explain below.

\subsection*{Dimension 3 example}

We take dimension 3 as an example, which contains the following (unnormalised) vectors in the standard complete MUB.
\begin{equation*}
\begin{bmatrix}
1 & 0 & 0 \\
0 & 1 & 0 \\
0 & 0 & 1 \\
\end{bmatrix}
\quad
\begin{bmatrix}
1 & 1 & 1 \\
1 & \omega & \omega^2 \\
1 & \omega^2 & \omega \\
\end{bmatrix}
\quad
\begin{bmatrix}
1 & \omega^2 & \omega^2 \\
\omega^2 & 1 & \omega^2 \\
\omega^2 & \omega^2 & 1 \\
\end{bmatrix}
\quad
\begin{bmatrix}
1 & \omega & \omega \\
\omega & 1 & \omega \\
\omega & \omega & 1 \\
\end{bmatrix} 
\label{eq:standard_set_3}
\end{equation*}
The MUBs in this set are related via the unitaries
\begin{equation}
U_F = \frac{1}{\sqrt{3}} \begin{pmatrix}
1 & 1 & 1 \\
1 & \omega & \omega^2 \\
1 & \omega^2 & \omega 
\end{pmatrix} 
\end{equation}
and
\begin{equation}
U_S = \begin{pmatrix}
1 & 0 & 0 \\
0 & \omega^2 & 0 \\
0 & 0 & \omega^2 
\end{pmatrix} .
\end{equation}
In $N=3$, there are two Alltop fiducials, given by Equation (\ref{eq:fiducial_3}). We start with the first one, i.e. $x=1$, and write down the 3 MUBs we obtain from the orbit under the WH group. Together with the computational basis, they form a complete MUB.
\begin{equation*}
\begin{bmatrix}
1 & 0 & 0 \\
0 & 1 & 0 \\
0 & 0 & 1 \\
\end{bmatrix}
\quad
\begin{bmatrix}
1 & 1 & 1 \\
\sigma & \sigma^4 & \sigma^7 \\
\sigma^2 & \sigma^8 & \sigma^5 \\
\end{bmatrix}
\quad
\begin{bmatrix}
1 & 1 & 1 \\
\sigma^7 & \sigma & \sigma^4 \\
\sigma^8 & \sigma^5 & \sigma^2 \\
\end{bmatrix}
\quad
\begin{bmatrix}
1 & 1 & 1 \\
\sigma & \sigma^4 & \sigma^7 \\
\sigma^8 & \sigma^5 & \sigma^2 \\
\end{bmatrix}
\label{eq:orbit_set1_3}
\end{equation*}
Note that the fiducial appears as the first vector in the second basis. Acting with $Z$ permutes vectors within a basis, while acting with $X$ permutes vectors between bases.
Similarly, the second fiducial in Equation (\ref{eq:fiducial_3}), i.e. $x=2$, produces three MUBs under the action of the WH group. Combined with the computational basis, they form another complete MUB.
\begin{equation*}
\begin{bmatrix}
1 & 0 & 0 \\
0 & 1 & 0 \\
0 & 0 & 1 \\
\end{bmatrix}
\quad
\begin{bmatrix}
1 & 1 & 1 \\
\sigma^2 & \sigma^5 & \sigma^8 \\
\sigma^4 & \sigma & \sigma^7 \\
\end{bmatrix}
\quad
\begin{bmatrix}
1 & 1 & 1 \\
\sigma^5 & \sigma^8 & \sigma^2 \\
\sigma^7 & \sigma^4 & \sigma \\
\end{bmatrix}
\quad
\begin{bmatrix}
1 & 1 & 1 \\
\sigma^2 & \sigma^5 & \sigma^8 \\
\sigma^7 & \sigma^4 & \sigma \\
\end{bmatrix}
\label{eq:orbit_set2_3}
\end{equation*}
The standard complete MUB and these two orbit complete MUBs overlap at the computational basis. We can find other non-overlapping orbit complete MUBs by acting with the $U_{F}$ and $U_{S}$ operators.
The bases in the standard complete MUB are related by the $U_F$ matrix and two applications of the $U_S$ matrix. If we apply the $U_F$ matrix to the two orbit complete MUBs given above, we find two more orbit complete MUBs. These then overlap at the second (Fourier) basis in the standard complete MUB. Likewise, if we apply the $U_S$ matrix to these we find two further orbit complete MUBs overlapping at the third basis in the standard complete MUB, and finally applying $U_S$ again produces two more orbit complete MUBs that overlap at the fourth basis in the standard complete MUB. In total for $N=3$ we have 8 orbit complete MUBs comprised from 24 individual bases plus the 4 in the standard complete MUB. 

We can cycle through overlapping complete MUBs using a unitary transformation. The unitary matrix is order 9 in dimension 3 and actually cycles through the vectors in each basis as well as the bases in different complete MUBs. For the three complete MUBs that overlap at the computational basis, the matrix is
\begin{equation}
\begin{pmatrix}
1 & & \\
 & \sigma & \\
 & & \sigma^2 \\
\end{pmatrix} .
\end{equation}
Starting from the first vector in the Fourier basis, the matrix has the effect
\begin{eqnarray}
\begin{pmatrix}
1 \\ 1 \\ 1
\end{pmatrix}
\rightarrow
\begin{pmatrix}
1 \\ \sigma \\ \sigma^2
\end{pmatrix}
\rightarrow
\begin{pmatrix}
1 \\ \sigma^2 \\ \sigma^4
\end{pmatrix}
\rightarrow
\begin{pmatrix}
1 \\ \omega\\ \omega^2
\end{pmatrix}
\rightarrow
\begin{pmatrix}
1 \\ \sigma^4 \\ \sigma^8
\end{pmatrix}
\rightarrow
\begin{pmatrix}
1 \\ \sigma^5 \\ \sigma
\end{pmatrix}
\rightarrow
\begin{pmatrix}
1 \\ \omega^2 \\ \omega
\end{pmatrix}
\rightarrow
\nonumber \\
\rightarrow
\begin{pmatrix}
1 \\ \sigma^7 \\ \sigma^5
\end{pmatrix}
\rightarrow
\begin{pmatrix}
1 \\ \sigma^8 \\ \sigma^7
\end{pmatrix}
\rightarrow
\begin{pmatrix}
1 \\ 1 \\ 1
\end{pmatrix}
\nonumber .
\end{eqnarray}
We see that it relates the vectors from one basis in each of the three complete MUBs in dimension 3.

\section{Orbit complete MUBs from the Clifford group}

Bandyopadhyay {\it et al.} showed that in Ivanovi\'{c}'s construction each basis in the standard complete MUB is an eigenbasis of a WH group element of the form $D_{\bf p}$ \cite{B}. In this section we make a similar statement about Alltop's construction, namely that each basis in an orbit complete MUB is an eigenbasis of a Clifford group element of the form $D_{\bf p}U_{\mathcal{G}}$, where $\mathcal{G}$ is order $N$.

While the vectors in the standard complete MUB are left invariant by an element of the WH group, the vectors in an orbit complete MUB are permuted under the action of the WH group: the operator $Z$ moves vectors within a basis and the operator $X$ moves vectors between bases. However, the vectors in the orbit complete MUB are left invariant by an element of the Clifford group. One way to see this is to look at the matrix components for the equation
\begin{equation}
\omega^k D_{ij} U_{\mathcal{S}} \ket{f_x} = \ket{f_x} . 
\label{eq:invariance}
\end{equation}
Using the following expression for the displacement operator
\begin{equation}
D_{ij} = \tau^{ij+2bj} \delta_{a,b+i} 
\end{equation}
together with Equations (\ref{eq:S}) and (\ref{eq:fiducial}) we find, for the case $N>3$,
\begin{equation}
\sum_{b,c} \omega^{k} \omega^{\frac{ij}{2}+bj} \delta_{a,b+i} \omega^{\frac{c^2}{2}} \delta_{b,c} \omega^{x c^3} = \omega^{x a^3}
\end{equation}
for some choice of $x$ in the Alltop fiducial. Solving this for the three parameters $i$, $j$ and $k$ will give the matrix that leaves $\ket{f_x}$ invariant. After summing, we find the solutions
\begin{equation}
i = \frac{1}{6x} \quad, \quad j = \frac{1}{12x} \quad , \quad
k = - \frac{1}{432x^2} .
\label{eq:solve_for_invariance}
\end{equation}
Recall $x$ is fixed by our choice of Alltop fiducial to investigate.
This shows that one particular Alltop fiducial is left invariant by a Clifford group element and consequently the remaining vectors in the orbit complete MUBs must also be left invariant by a Clifford group element. The Clifford group contains the WH group as a subgroup, so the orbit and standard complete MUBs are now on somewhat equal footing in terms of Clifford group invariance.

In Section 2, we partitioned the Clifford group into elements that could be written as WH translates and those that could not. The bases in the orbit complete MUBs are left invariant by an order $N$ Clifford element of the second type. We can see this by comparing the number of bases in the orbit complete MUBs with the number of order $N$ subgroups of the second type of Clifford element. We showed in Section 2 that there are $N^2(N+1)$ order $N$ subgroups in the Clifford group and $N(N+1)$ subgroups of translates, which leaves $N(N^2-1)$ subgroups that cannot be written in translate form. This is precisely the number of bases in the orbit complete MUB construction. So we conclude that every order $N$ Clifford element that cannot be written as a WH translate has an eigenbasis that is a basis in an orbit complete MUB.

This is in direct analogy to the standard complete MUB, where the bases are eigenbases of subgroups of the WH group. Here, the orbit complete MUBs contain bases that are eigenbases of order $N$ subgroups of the Clifford group.

\subsection*{Dimension 3 example}

We can look at how this works in dimension 3. First, we want to know which Clifford element leaves the fiducial vectors invariant. In this example, the method outlined above runs into difficulties, but as the Clifford group is fairly manageable in dimension 3 we can search directly. The first fiducial vector is invariant under the Clifford element 
\begin{equation}
U_{\mathcal{G}} = \sigma^2 X^2 S = 
\begin{pmatrix}
0 & \sigma^8 & 0 \\
0 & 0 & \sigma^8 \\
\sigma^2 & 0 & 0 \\
\end{pmatrix} .
\end{equation}
The second fiducial is invariant under the Clifford element 
\begin{equation}
U_{\mathcal{G}} = \sigma^4 X^4 S^2 = 
\begin{pmatrix}
0 & \sigma^7 & 0 \\
0 & 0 & \sigma^7 \\
\sigma^4 & 0 & 0 \\
\end{pmatrix}  .
\end{equation}

In $N=3$, there are 72 Clifford elements (not counting WH group elements). Of these, 24 can be written as WH translates, leaving 48 that cannot. These latter elements collect into 24 subgroups of order 3, which matches the number of individual bases in the 8 orbit complete MUBs.

\section{Distances between bases}

Given all these MUBs we can ask where they sit relative to one another. To do this, we consider each basis as a point in a Grassmannian space and compute distances between these points. We first introduce the so-called chordal Grassmannian distance developed in \cite{distance, Ingemar} and then use it to prove two theorems about the distance between bases from different complete MUBs.

As an alternative to Hilbert space, we can picture quantum state space in terms of density operators in what is sometimes called Bloch space: the space of all hermitian matrices with unit trace. This is an $(N^2-1)$--dimensional space which we can treat as a vector space with the maximally mixed state $\rho_{*}=\frac{1}{N} \mathbbm{1}$ at the origin. The set of density matrices forms a convex body in this space, while the projectors corresponding to pure states span a $2(N-1)$--dimensional continuous subspace. A vector in this space is then a traceless matrix, whose explicit construction is given by
\begin{equation}
{\bf e} = \ket{e} \bra{e} - \rho_{*} .
\end{equation}

An orthonormal basis in Hilbert space corresponds to a regular $(N-1)$--simplex in Bloch space and each simplex spans a unique $(N-1)$--plane through the origin.
For two MUBs the two planes are totally orthogonal, so the equations
\begin{equation}
| \left\langle e_i | f_j \right\rangle |^2 = \frac{1}{N}
\end{equation}
and
\begin{equation}
{\bf e}_i \cdot {\bf f}_j = 0
\end{equation}
express the same condition, where $\ket{e_i}$ and $\ket{f_j}$ are vectors in Hilbert space and ${\bf e}_i$ and ${\bf f}_j$ are their corresponding vectors in Bloch space. By totally orthogonal planes we mean that every vector in one plane is orthogonal to every vector in the other plane. This leads to a nice geometrical argument for the upper limit of $N+1$ MUBs in a complete MUB; since the dimension of Bloch space is $N^2-1=(N+1)(N-1)$ we can't fit more than $N+1$ totally orthogonal planes in it. 

There is a natural way to compute distances between planes in Bloch space by considering them as points in a Grassmannian space. We perform the same trick as when we moved from Hilbert space to Bloch space, where the vectors are now formed from projectors onto the $(N-1)$--planes. We may then use the squared chordal Grassmannian distance \cite{distance, Ingemar}, given in terms of our Hilbert space vectors by
\begin{equation}
D^{2}_{c} = 1 - \frac{1}{N-1} \sum_{i=0}^{N-1} \sum_{j=0}^{N-1} \left( | \left\langle e_i | f_j \right\rangle |^2 - \frac{1}{N} \right)^2 .
\label{eq:distance}
\end{equation}
It is called chordal because the $(N-1)$--planes can be considered as points lying on a sphere in a high dimensional Euclidean space (see \cite{Ingemar} for more details). The distance has the range
\begin{equation}
0 \leq D^{2}_{c} \leq 1 
\end{equation}
where the minimum occurs for identical bases and the maximum for MUBs. This concept of distance was used to search for MUBs in dimension 6 by maximising the value of $D^{2}_{c}$ between bases \cite{Berge}.

We can ask for the average distance between two bases in this space. It is calculated in \cite{Ingemar} by taking the computational basis and one chosen at random according to the Fubini-Study measure and results in an average distance of
\begin{equation}
\left\langle D^{2}_{c} \right\rangle = \frac{N}{N+1} .
\label{eq:average}
\end{equation}

We need one more result before we look at the distances between bases: complete MUBs are 2-designs \cite{Barnum, KR1}. Given a function on $\mathbbm{C}^N$ that is homogeneous of order 2 in both its coordinates and their complex conjugates, averaging over a particular set of projective points will give the same value as averaging over the whole space. The particular set of points for which this happens is called a 2-design. If we fix one basis then the chordal Grassmannian distance in Equation (\ref{eq:distance}) is a function of the above form. Evaluating the distance over a complete MUB, or any other 2-design, will then give the same result as the average distance in Equation (\ref{eq:average}). We discuss the two distances separately.

\begin{theorem}
Two distinct bases lying in overlapping complete MUBs, either from the standard or orbit construction, are separated by a distance of
\begin{equation*}
D^{2}_{c,1} = \frac{N-1}{N} .
\end{equation*}
\end{theorem}

\begin{proof}
We proceed in two steps. Firstly, we prove that the distance from one basis in the standard complete MUB to every basis in an orbit complete MUB is the same (not including the shared basis). Secondly, we prove that this distance is the same regardless of which pair of overlapping complete MUBs we choose.

In step one, we need to calculate the distance between a basis in the standard complete MUB and a basis in an orbit complete MUB. From Equation (\ref{eq:distance}), this requires $N^2$ scalar products for which we might expect $N^2$ distinct values, but we find only $N$. This is because $Z$ relates vectors within a basis for both complete MUBs, so all scalar products can be expressed in terms of just one vector from the standard complete MUB basis.\footnote{Alternatively, we could have rewritten the scalar products in terms of just one basis vector in the orbit complete MUB.} For example, if the standard complete MUB basis is labelled by the vectors $\left\{ \ket{s_i} \right\}_{1}^{N}$, related by $Z \ket{s_i} = \ket{s_{i+1}}$, and the orbit complete MUB basis is labelled by $\left\{ \ket{a_i} \right\}_{1}^{N}$ and similarly related by $Z \ket{a_i} = \ket{a_{i+1}}$ (all arithmetic modulo $N$), then we can rewrite the following scalar products as
\begin{align}
\sprod{s_1}{a_1} \quad , \quad &\sprod{s_2}{a_1} = \scprod{s_1}{Z^{\dag}}{a_1} = \sprod{s_1}{a_N} \quad , \quad \ldots
\notag \\
\sprod{s_1}{a_2} \quad , \quad &\sprod{s_2}{a_2} = \scprod{s_1}{Z^{\dag}}{a_2} = \sprod{s_1}{a_1} \quad , \quad \ldots
\label{eq:scalprod1} \\
\sprod{s_1}{a_3} \quad , \quad &\sprod{s_2}{a_3} = \scprod{s_1}{Z^{\dag}}{a_3} = \sprod{s_1}{a_2} \quad , \quad \ldots
\notag 
\end{align}
and so on. They reduce to only $N$ possibilities, all in terms of the vector $\ket{s_1}$.

Turning to the distance between the standard complete MUB basis and a second basis in the orbit complete MUB, we can play a similar trick. Let the vectors in the second orbit complete MUB basis be labelled by $\left\{ \ket{b_i} \right\}_{0}^{N-1}$ and again related by $Z \ket{b_i} = \ket{b_{i+1}}$. As above, we find only $N$ distinct scalar products of the form
\begin{align}
\sprod{s_1}{b_1} \quad , \quad &\sprod{s_2}{b_1} = \scprod{s_1}{Z^{\dag}}{b_1} = \sprod{s_1}{b_N} \quad , \quad \ldots
\notag \\
\sprod{s_1}{b_2} \quad , \quad &\sprod{s_2}{b_2} = \scprod{s_1}{Z^{\dag}}{b_2} = \sprod{s_1}{b_1} \quad , \quad \ldots
\\
\sprod{s_1}{b_3} \quad , \quad &\sprod{s_2}{b_3} = \scprod{s_1}{Z^{\dag}}{b_3} = \sprod{s_1}{b_2} \quad , \quad \ldots
\notag 
\end{align} 
and so on.

Recall that $X$ relates vectors within a basis in the standard complete MUB and vectors between bases in the orbit complete MUBs. For argument's sake, say $X \ket{s_i} = \ket{s_{i+1}}$ and $X \ket{a_i} = \ket{b_i}$. We can therefore express the scalar products involving the second orbit basis in terms of the first orbit basis as follows
\begin{align}
\sprod{s_1}{b_1} = \scprod{s_1}{X}{a_1} = \sprod{s_N}{a_1} 
= \scprod{s_1}{(Z^{\dagger})^{N-1}}{a_1} = \sprod{s_1}{a_2}
\notag \\
\sprod{s_1}{b_2} = \scprod{s_1}{X}{a_2} = \sprod{s_N}{a_2}
= \scprod{s_1}{(Z^{\dagger})^{N-1}}{a_2} = \sprod{s_1}{a_3}
\label{eq:scalprod2} \\
\sprod{s_1}{b_3} = \scprod{s_1}{X}{a_3} = \sprod{s_N}{a_3}
= \scprod{s_1}{(Z^{\dagger})^{N-1}}{a_3} = \sprod{s_1}{a_4}
\notag 
\end{align}
and so on for the remaining basis vectors and similarly for the other bases in the orbit complete MUB. The second part of the expressions in Equation (\ref{eq:scalprod2}) follow the procedure in Equation (\ref{eq:scalprod1}). Bases in the standard complete MUB are left invariant by $Z$ and $X$, but the arguments given here still apply. 
The scalar products for any two bases (one from the standard complete MUB and one from an orbit complete MUB) can always be rewritten to give the same $N$ values. Subsequently, the distance between any two bases, calculated from summing the scalar products, must be the same.
The complete MUBs are unitarily equivalent so the argument applies equally to two overlapping orbit complete MUBs.

In step two, we use the 2-design property to show that the same distance arises for all pairs of overlapping complete MUBs. The average distance between a fixed basis in one complete MUB and every basis in an overlapping complete MUB must obey
\begin{equation}
1 + N D^{2}_{c,1} = (N+1) \left\langle D^{2}_{c} \right\rangle .
\end{equation}
The first term is the contribution from the basis at which the two complete MUBs overlap (it belongs to both complete MUBs so has a distance equal to 1); the second term is the contribution from the $N$ remaining bases; the final term is the average distance between bases. This leads to
\begin{equation}
D^{2}_{c,1} = \frac{N-1}{N} 
\end{equation}
for the distance between two bases from overlapping complete MUBs.
\end{proof}

We now turn to the second distance, which we prove in a very similar manner.

\begin{theorem}
Two distinct bases lying in non-overlapping orbit complete MUBs are separated by a distance of
\begin{equation*}
D^{2}_{c,2} =  \frac{N-1}{N} - \frac{1}{N^2} .
\end{equation*}
\end{theorem}

\begin{proof}
We proceed as before. In step one we show that the distance between a fixed basis from an orbit complete MUB and every other basis from a second orbit complete MUB is always constant. In step two we show that this distance is the same for every pair of orbit complete MUBs.

In step one, we again find that the scalar products between two bases give $N$ distinct values. If we look at a vector in the fixed basis, it will be invariant under an element of the Clifford group $D_{\bf p} U_{\mathcal{G}}$. Acting with this element will permute vectors between bases in the second orbit complete MUB and we can rewrite things to find the same $N$ scalar products between any two bases, as in the previous proof. Consequently, the distance between any two bases from orbit complete MUBs is constant.

In step two, we show that the distance takes the same value for every pair of orbit complete MUBs. Using the 2-design property once more, we know the average distance between a fixed basis in one complete MUB and every basis in a non-overlapping complete MUB must obey
\begin{equation}
D^{2}_{c,1} + N D^{2}_{c,2} = (N+1) \left\langle D^{2}_{c} \right\rangle .
\end{equation}
The first term is the contribution from the basis that lies in two overlapping complete MUBs (i.e. a basis from the standard complete MUB); the second term is the contribution from the $N$ remaining bases; the final term is the average distance between bases. This leads to
\begin{equation}
D^{2}_{c,2} = \frac{N-1}{N} + \frac{1}{N^2}
\end{equation}
for the distance between two bases from non-overlapping orbit complete MUBs.

\end{proof}

\section{Linear dependencies among bases}

The question of whether vectors in WH orbits are linearly dependent is important for signal processing applications where signals are transmitted over lossy channels. It was first proved that a WH orbit always exists where any subset of $N$ vectors are linearly independent for prime dimensions \cite{Pfander} and later for arbitrary dimension \cite{singapore}. Some WH orbits exhibit a high number of linear dependences when the fiducial vectors have certain symmetries \cite{Hoan}.

The complete MUBs in this paper contain many linearly dependent vectors. One quick way of generating dependencies is to use the phase-point operators introduced by Wootters \cite{Wootters}. They include the unitary representation of the symplectic element
\begin{equation}
\mathcal{A} = 
\left(
\begin{array}{cc}
-1 & 0 \\
0 & -1
\end{array}
\right)
\end{equation}
plus its WH translates. The unitary $U_{\mathcal{A}}$ has eigenvalues $+1$ and $-1$ with corresponding eigenspaces of dimension $\frac{1}{2}(N+1)$ and $\frac{1}{2}(N-1)$. One vector from each basis in the standard complete MUB is invariant under the action of this phase-point operator and so lies in the former eigenspace, resulting in $N+1$ linearly dependent vectors. Unitary equivalence ensures there are linearly dependent vectors in the orbit complete MUBs, too.

We focus here on linear dependencies that involve both types of complete MUBs simultaneously. In dimension $N=3k+2$ we find that the $(N-1)$--planes spanned by some sets of linearly dependent vectors in the standard complete MUB are orthogonal to vectors in the orbit complete MUBs. This doesn't happen in dimension $N=3k+1$.

\begin{theorem}
For dimensions $N=3k+2$, we can find a set of $N$ linearly dependent vectors, where each vector lies in a different basis in the standard complete MUB, that is orthogonal to a vector from an orbit complete MUB.
\end{theorem}

\begin{proof}
We look first at a specific case, namely the vector $\ket{u}$ whose entries are all $\frac{1}{\sqrt{N}}$. It lies in the Fourier basis in the standard complete MUB. Taking the scalar product with an Alltop fiducial vector gives

\begin{equation}
\left\langle u | f_x \right\rangle = \frac{1}{N} \sum_{a=0}^{N-1} \omega^{a^{3}}.
\end{equation}

The two vectors are orthogonal when the sum of the roots of unity vanishes. This requires all terms must be distinct, i.e. $\omega^{a^{3}}$ should never equal $\omega^{b^{3}}$ when $a \neq b$. The relevant equation is
\begin{equation}
a^3 = b^3 \mbox{ mod } N \quad \Leftrightarrow \quad \left( \frac{a}{b} \right)^3 = 1 .
\end{equation}
Substituting $x$ for $\frac{a}{b}$, we need to solve the equation
\begin{equation}
x^3 - 1 = 0 
\end{equation}
which, after rewriting and completing the square, leads to
\begin{equation}
x = \frac{\pm \sqrt{-3} - 1}{2} .
\end{equation}
So the question becomes whether $-3$ is a quadratic residue modulo $N$. In dimension $N=3k+2$ the answer is no \cite{Apostol} and therefore we satisfy the condition $\left\langle u | f_x \right\rangle = 0$.

The Alltop fiducial is invariant under something in the Clifford group of the form $D_{\bf p} U_{\mathcal{S}}$. Acting with this on the vector $\ket{u}$ generates $N$ vectors, each from a different basis in the standard complete MUB (excluding the basis at which the standard and orbit complete MUBs overlap). In dimension $N=3k+2$, these $N$ vectors all lie orthogonally to the Alltop fiducial and are therefore linearly dependent.

\end{proof}

\section{Orbit complete MUBs in even prime power dimensions}

Just as Ivanovic's construction was generalised to the prime power case by Wootters and Fields \cite{WF}, the construction based on Alltop's functions was extended to odd prime power dimensions by Klappenecker and R\"{o}tteler \cite{KR2}. In this section, we give explicit expressions for a fiducial vector that lead to orbit complete MUBs in the two lowest even prime power dimensions.

In prime power dimensions, we have a choice of WH groups. The one we have been using, $H(p)$, can be generalised in two ways, either to
\begin{equation}
H(p^n) \quad \mbox{or} \quad H(p) \times \ldots \times H(p) .
\end{equation}
The second group, sometimes called the extraspecial Heisenberg group \cite{H_group}, is relevant for the MUB problem.
Considered projectively, it still has $N^2$ elements but they are now obtained by taking tensor products of $\mathbbm{1}$, $X$ and $Z$. For $N=2$, this is identical to the WH group, but in higher dimensions it produces a different group.

For $N=2$, the standard complete MUB is well-known and contains the following (unnormalised) vectors.
\begin{equation*}
\begin{bmatrix}
1 & 0 \\
0 & 1 
\end{bmatrix}
\quad
\begin{bmatrix}
1 & 1 \\
1 & -1 
\end{bmatrix}
\quad
\begin{bmatrix}
1 & 1 \\
i & -i
\end{bmatrix} 
\label{eq:standard_set_2}
\end{equation*}
The bases are related by the Fourier matrix, see Equation (\ref{eq:fourier}), and the order 4 unitary matrix
\begin{equation}
\begin{pmatrix}
1 & 0 \\
0 & i
\end{pmatrix} .
\end{equation}
We now ask whether complete MUBs exist that are formed from orbits under the extraspecial Heisenberg group. The answer is yes: acting on the fiducial vector
\begin{equation}
\ket{f} = \left( 1 , \mu \right)^T 
\end{equation}
where $\mu = e^{\frac{i \pi}{4}}$ produces an orbit complete MUB (when combined with the computational basis).
\begin{equation*}
\begin{bmatrix}
1 & 0 \\
0 & 1 
\end{bmatrix}
\quad
\begin{bmatrix}
1 & 1 \\
\mu & \mu^5 
\end{bmatrix}
\quad
\begin{bmatrix}
1 & 1 \\
\mu^7 & \mu^3
\end{bmatrix} 
\label{eq:orbit_set1_2}
\end{equation*}
As in the prime-dimensional case, we can obtain two further orbit complete MUBs by acting with the operators that relate the bases in the standard complete MUB. This gives three orbit complete MUBs in total.

For $N=4$, we have the two fiducial vectors
\begin{equation}
\ket{f_1} = \left( 1,1,1,i \right)^T
\end{equation}
and
\begin{equation}
\ket{f_2} = \left( 1,1,1,-i \right)^T 
\end{equation}
whose orbits under the extraspecial Heisenberg group produce two orbit complete MUBs (when each is combined with the computational basis). We obtain eight further orbit complete MUBs by acting with the same operators that relate the bases in the standard complete MUB. We therefore find ten orbit complete MUBs in total.

For $N=8$, we performed a computer search for fiducials whose components are sixteenth roots of unity but were unable to find any that led to an orbit complete MUB. We leave the question open as to whether a generalisation can be found for dimensions $N=2^k$.

\section{Summary}

We looked at Alltop's complete MUBs---sets of $N+1$ MUBs---in prime dimensions and investigated the role of finite groups in their construction. Specifically, we showed how Alltop's construction leads to several complete MUBs as orbits under the WH group and how these orbit complete MUBs can be generated using abelian subgroups of the Clifford group. This approach removes the reliance on Alltop's fiducial vectors and reveals a structure analogous to the standard complete MUB construction, which uses abelian subgroups of the WH group.

From a geometrical perspective each basis in a complete MUB spans an $(N-1)$--dimensional plane in the space of hermitian matrices with unit trace. We calculated distances between these planes, proving that two planes are separated by a distance of $D^{2}_{c,1} = \frac{N-1}{N}$ when they correspond to bases from overlapping complete MUBs (Theorem 1) and $D^{2}_{c,2} = \frac{N-1}{N} - \frac{1}{N^2}$ when they correspond to bases from non-overlapping complete MUBs (Theorem 2).

We showed that in dimension $N=3k+2$ we can always find a linearly dependent set of $N$ vectors from distinct bases in the standard complete MUB that is orthogonal to a vector in an orbit complete MUB (Theorem 3). Finally, we gave explicit expressions for fiducial vectors of orbit complete MUBs in dimensions $N=2$ and 4. As far as we know, higher dimensional orbit complete MUBs in even prime power dimensions are unknown.

\section*{Acknowledgements}

I am grateful to Ingemar Bengtsson for useful advice and comments. Also, I thank Mark Howard for drawing my attention to the connection to quantum computing in Ref. \cite{Howard} and for sharing his work related to the end of Section 3. Finally, thanks go to two anonymous referees who helped improve this paper.

\section*{Appendix: Orbit complete MUBs and SICs}

The introduction of the Clifford group to the construction of orbit complete MUBs sheds light on a recent observation about symmetric informationally complete positive operator valued measures (SICs). Specifically, it answers a question about unitary equivalence among SICs in dimension 3. A SIC \cite{Marcus, Zauner, Caves, SG} is a collection of $N^2$ unit vectors that obey the equation
\begin{equation}
| \left\langle e_i | e_j \right\rangle |^2 = \frac{1}{N+1} \quad , \quad i \neq j.
\label{eq:SIC}
\end{equation}
There is no general proof of SIC existence, but a well-known conjecture states that in every dimension we can find a fiducial vector, invariant under an order 3 Clifford unitary, whose WH orbit produces a SIC \cite{Marcus, Zauner}. Although this conjecture holds in every dimension so far tested (which by now is around 80 dimensions), it remains a puzzling observation.

These order 3 unitaries appear in the orbit complete MUBs construction, too. As a consequence of the orbit complete MUBs lying in different Clifford orbits in dimension $N = 3k + 1$ \cite{HC}, the Alltop fiducial vectors in Equation (\ref{eq:fiducial}) must be left invariant by more Clifford unitaries than the one used in Equation (\ref{eq:invariance}). The additional unitaries that leave the fiducial invariant are the order 3 elements in the SIC problem. In these dimensions, the Alltop fiducials lie in the same subspaces as the SIC fiducials.

Dimension 3 holds a continuous family of SICs, parametrised by the angle $\phi$ in the fiducial vector
\begin{equation}
\frac{1}{\sqrt{2}}
\left( 
\begin{array}{c}
0 \\
1 \\
-e^{i \phi} \\
\end{array}
\right) .
\end{equation}
The usual way to characterise SICs is to partition them into orbits under the extended Clifford group (the Clifford group plus all anti-unitaries that stabilise the WH group), where, normally, SICs on different orbits are not unitarily equivalent. Up to transformations by the extended Clifford group, it is enough to consider SIC fiducials with $\phi \in \left[ 0 , \frac{2 \pi}{6} \right]$. SICs then lie on orbits of three different lengths: if $\phi=0$ the orbit contains only one SIC, if $\phi=\frac{2 \pi}{6}$ the orbit contains four SICs, and for every other $\phi$ the orbit contains eight SICs \cite{Marcus}. Recently, an unexplained unitary relation between SICs with $\phi$ and $\phi + \frac{2 \pi}{9}$ was found \cite{Zhu}. The explicit form of the unitary transformation is given by
\begin{equation}
U = 
\begin{pmatrix}
1 &  &  \\
 & \sigma^8 &  \\
 &  & \sigma^7 \\
\end{pmatrix}
\label{eq:Zhu}
\end{equation}
where $\sigma = e^{\frac{2 \pi i}{9}}$.

Before we show that this unitary relation appears naturally from the standpoint of the orbit complete MUBs, we first highlight the correspondence between SICs and MUBs in dimension 3. The SIC originating from the fiducial with $\phi=0$ can be obtained from the inflection points of one of the elliptic curves in a family of curves in the complex projective plane called the Hesse pencil \cite{Hughston, Ingemar-Eddington, Hoan}. This leads to a combinatoric structure called the Hesse configuration which singles out twelve vectors in Hilbert space. These are the vectors in the standard complete MUB. The eight SICs on the orbit containing the fiducial with $\phi=\frac{2 \pi}{9}$ are also related to an elliptic curve in the Hesse pencil and they correspond to eight further sets of twelve vectors via the Hesse configuration. These eight sets are the eight orbit complete MUBs.

It is then apparent that there must be a unitary transformation not in the Clifford group that relates the SICs with $\phi=0$ and $\phi = \frac{\pi}{9}$ because they correspond to complete MUBs---the standard one and the eight orbit ones, respectively---that are unitarily equivalent but lie in different Clifford orbits. Note that the unitary transformation between SICs given in Equation (\ref{eq:Zhu}) is precisely the unitary transformation that relates the standard complete MUB and the second orbit complete MUB in dimension 3, i.e. the diagonal components are the components from the fiducial with $x=2$ in Equation (\ref{eq:fiducial_3}).

\end{document}